\documentclass[12pt]{article}

\usepackage{amsmath,amssymb,amsfonts,amsthm,mathtools}
\usepackage{mathabx}
\usepackage{stmaryrd}
\usepackage{wasysym}

\usepackage{graphicx}
\usepackage{epstopdf}
\usepackage[table]{xcolor}
\usepackage{adjustbox}

\usepackage{array}
\usepackage{arydshln}
\usepackage{booktabs}
\usepackage{caption}
\usepackage{float}
\usepackage{longtable}
\usepackage{tablefootnote}

\usepackage{enumitem}
\usepackage{multicol}
\usepackage{indentfirst}
\usepackage{anyfontsize}

\usepackage{textcomp}
\usepackage{verbatim}
\usepackage[misc]{ifsym}
\usepackage{algorithmic}
\usepackage{footmisc}
\usepackage{xpatch}

\usepackage[numbers]{natbib}
\usepackage{xurl}
\usepackage[colorlinks=true]{hyperref}
		\hypersetup{urlcolor=blue, citecolor=red}
\usepackage{cleveref}
\usepackage{aliascnt}

\makeatletter
\AtBeginDocument{%
  \xpatchcmd{\@thm}{\thm@headpunct{.}}{\thm@headpunct{}}{}{}%
}
\makeatother

\newtheorem{theorem}{Theorem}[section]

\newaliascnt{lemma}{theorem}

\aliascntresetthe{lemma}

\newaliascnt{proposition}{theorem}

\aliascntresetthe{proposition}

\newaliascnt{corollary}{theorem}

\aliascntresetthe{corollary}

\newaliascnt{conjecture}{theorem}

\aliascntresetthe{conjecture}

\crefname{theorem}{theorem}{theorems}
\Crefname{theorem}{Theorem}{Theorems}

\crefname{lemma}{lemma}{lemmas}
\Crefname{lemma}{Lemma}{Lemmas}

\crefname{proposition}{proposition}{propositions}
\Crefname{proposition}{Proposition}{Propositions}

\crefname{corollary}{corollary}{corollaries}
\Crefname{corollary}{Corollary}{Corollaries}

\crefname{conjecture}{conjecture}{conjectures}
\Crefname{conjecture}{Conjecture}{Conjectures}

\theoremstyle{definition}

\newaliascnt{definition}{theorem}
\newtheorem{definition}[definition]{Definition}
\aliascntresetthe{definition}

\newaliascnt{example}{theorem}
\newtheorem{example}[example]{Example}
\aliascntresetthe{example}

\newaliascnt{remark}{theorem}

\aliascntresetthe{remark}

\crefname{definition}{definition}{definitions}
\Crefname{definition}{Definition}{Definitions}

\crefname{example}{example}{examples}
\Crefname{example}{Example}{Examples}

\crefname{remark}{remark}{remarks}
\Crefname{remark}{Remark}{Remarks}

\DeclareCaptionType{equ}[][]
\theoremstyle{plain}

\renewcommand{\thefootnote}{\fnsymbol{footnote}}

\newcommand{\doititle}[1]{#1}

\allowdisplaybreaks

\begin{document}
\title{Some MDS and ACD codes over commutative non-unital rings of orders 4 and 9 (Revision){\color{red}\footnotemark[1]} }

\author{Jon-Lark Kim \\ Department of Mathematics and \\ Institute for Mathematical and Data Sciences \\ Sogang University, Seoul, Korea \\
		{\tt jlkim@sogang.ac.kr } \\
             \\ Marvin Olavides \\ Department of Mathematics and \\ Institute for Mathematical and Data Sciences \\ Sogang University, Seoul, Korea \\
             {\tt mmolavides@gmail.com}\\
             \\ Young Gun Roe{\color{red}\footnotemark[2]} \\ KNU Research Institute for Mathematical Sciences \\ Kangwon National University, Chuncheon, Korea \\ 
		{\tt ygroe@naver.com}
	}

\date{ }

\maketitle

\begingroup
\renewcommand{\thefootnote}{\fnsymbol{footnote}}

\footnotetext[1]{This is a correction of the paper published in \textit{Advances in Mathematics of Communications}, Volume 24, pages 61–76, 2026. In particular, we corrected the statements of Theorems~\ref{thm: ACD iff LCD I2-code}, \ref{thm: ACD iff LCD I3-code}, \ref{thm: ACD iff LCD Ip-code}, and their proofs.}

\footnotetext[2]{Corresponding author}
\endgroup


\begin{abstract}
There are eleven finite rings of order $p^{2}$ denoted by $A_p$ to $K_p$ in alphabetical order. In particular, we consider $I_{2}$ and $I_{3}$ which are commutative non-unital rings of orders 4 and 9 defined by generators and relations as
\[I_{p}=\left\langle a,b\mid pa=pb=0,\:a^{2}=b,\:ab=0\right\rangle\] for $p=2, 3,$ respectively.
Alahmadi et al. studied codes over these rings. 
In this paper, we study additive complementary dual (ACD) codes over the rings $I_{2}$ and $I_{3}$. We show relations between ACD codes over $I_{2}$ and binary linear complementary dual (LCD) codes using a reduction map from $I_{2}$ to $\mathbb{F}_{2}$, and between ACD codes over $I_{3}$ and ternary LCD codes using a reduction map from $I_{3}$ to $\mathbb{F}_{3}$. Using the first relation, we classify ACD codes over $I_{2}$ with the highest minimum distances for $n=1, 2, 3$ and partially for $n=4, 5$. It turns out that they are maximum distance separable (MDS) codes.
 Using the second relation, we classify ACD codes over $I_{3}$ with the highest minimum Lee distances for $n=1, 2$ and partially for $n=3$. We generalize the two relations into a relation between ACD codes over $I_{p}$ and $p$-ary LCD codes using a reduction map from $I_{p}$ to $\mathbb{F}_{p}$.
\end{abstract}

{\bf{Keywords}} : additive codes, LCD codes, non-unital ring

{\bf{Mathematics Subject Classification}} : 94B05, 16L99

\section{Introduction}
In the early stage of coding theory history, only codes over binary field were considered. Soon after the alphabet was generalized to finite fields. In the 1990s, a connection between linear codes over $\mathbb{Z}_{4}$ and non-linear binary codes was found \cite{Sole4}. After this work was known, many papers on codes over $\mathbb{Z}_{4}$ were published. Then the interest was shifted to codes over commutative rings of order 4 \cite{Dougherty}. Recently, commutative non-unital rings of orders 4 and 9, denoted by $I_{2}$ and $I_{3}$ in the classification of
\cite{Fine}, began to be used as the alphabet \cite{Sole1}. The rings $I_{2}$ and $I_{3}$ are commutative non-unital rings of orders 4 and 9 defined by generators and relations as
\[I_{p}=\left\langle a,b\mid pa=pb=0,\:a^{2}=b,\:ab=0\right\rangle\] for $p=2, 3$ respectively. Alahmadi et al. \cite{Sole1} introduced quasi Type IV codes (quasi self-dual codes with even torsion code) over the ring $I_{2}$. Kim et al. \cite{Roe1} constructed more quasi self-dual codes over $I_{2}$. According to \cite{Kim1}, codes over $I_{2}$ have applications in constructing DNA codes. Alahmadi et al. \cite{Sole9} investigated building up constructions for codes over the ring $I_{3}$.

A linear complementary dual (LCD) code was first introduced by Massey \cite{Massey1} as a reversible code in 1964. Since then, many papers on LCD codes were published and LCD codes were applied in many areas such as cryptography, communication systems and data storage. Massey \cite{Massey2} found that there exist asymptotically good LCD codes. Yang et al. \cite{Massey3} showed a necessary and sufficient condition for a cyclic code to have a complementary dual. Li et al. \cite{Li} studied a family of BCH codes and extended their study to LCD BCH codes. Tzeng \cite{Tzeng} showed that a class of reversible codes has minimum distance greater than that given by the BCH bound. Dougherty et al. \cite{Kim4} gave a linear programming bound on the largest size of an LCD code. Galvez et al. \cite{Roe2} followed this line to find some bounds on binary LCD codes. In 2014, Carlet et al. \cite{Carlet} constructed LCD codes in several ways, and found an application of LCD codes in cryptography against side-channel attacks and fault injection attacks. Thus, it is natural to study LCD codes.

In Carlet's Boolean masking approach, the minimum distance of a code is a performance criterion, and this approach uses only the additivity of codes \cite{Kim2}, \cite{Kim3}. Additive codes are also used in quantum error-correction and quantum computing \cite{Dougherty2}. Moreover, since additive codes include linear codes, additive complementary dual (ACD) codes form a larger class that contain LCD codes. Here we note that according to \cite{Carlet2}, linear codes over $\mathbb{F}_{q}$ are equivalent to LCD codes for $q>3$. Thus, much efforts have been made in the literature to study LCD codes, but relatively less effort has been made for ACD codes. However, some papers on ACD codes have been published recently. Shi et al. \cite{Kim3} studied ACD codes over the ring $E$, and said that ACD codes form a natural generalization of LCD codes over $E$, and the minimum distances of ACD codes over $E$ seem as good as those of Hermitian LCD codes over $\mathbb{F}_{4}$. In \cite{Kim2}, Shi et al. investigated ACD codes over $\mathbb{F}_{4}$. They showed that ACD codes over $\mathbb{F}_{4}$ are sometimes better than LCD codes over $\mathbb{F}_{4}$. They also said that the application of ACD codes to security still makes sense. Shi et al. \cite{Sole5} further studied cyclic ACD codes over $\mathbb{F}_{4}$.  Dougherty et al. \cite{Dougherty2} examined ACD codes from group characters. They defined ACD codes over a finite abelian group, and showed that the best minimum weight of ACD codes is always greater than or equal to the best minimum weight of LCD codes of the same size, and the inequality is often strict. Choi et al. \cite{Kim6} made some observations on MDS subclass of ACD codes. Furthermore, as we shall see later in our paper, there is no LCD code over the rings $I_{2}$ and $I_{3}$, which naturally makes us turn our attention to ACD codes.

Motivated by the above two aspects, we consider ACD codes over $I_{2}$ and $I_{3}$. We show relations between ACD codes over $I_{2}$ and binary LCD codes using a reduction map from $I_{2}$ to $\mathbb{F}_{2}$, and between ACD codes over $I_{3}$ and ternary LCD codes using a reduction map from $I_{3}$ to $\mathbb{F}_{3}$. Using the first relation and Magma computation, we classify ACD codes over $I_{2}$ with the highest minimum distances for $n=1, 2, 3$ and partially for $n=4, 5$.
Using the second relation and Magma computation, we classify ACD codes over $I_{3}$ with the highest minimum distances for $n=1, 2$ and partially for $n=3$. We generalize the two relations into a relation between ACD codes over $I_{p}$ and $p$-ary LCD codes using a reduction map from $I_{p}$ to $\mathbb{F}_{p}$. To be precise, the relation says that an additive code $\mathcal{C}\subseteq I_{p}^{n}$ is ACD if and only if $\beta_{p}\left(\mathcal{C}\right)\subseteq\mathbb{F}_{p}^{n}$ is LCD where $\beta_{p}$ is the reduction map and $| \beta_{p} ( \mathcal{C} ) | = | \mathcal{C} |$. This relation is proved by showing that $\mathcal{C}\cap\mathcal{C}^{\perp}=\left\{ \mathbf{0}\right\} $ if and only if $\beta\left(\mathcal{C}\right)\cap\left(\beta\left(\mathcal{C}\right)\right)^{\perp}=\left\{ \mathbf{0}\right\}$, and $\left|\mathcal{C}\right|\left|\mathcal{C}^{\perp}\right|=p^{2n}$ if and only if $\left|\beta_{p}\left(\mathcal{C}\right)\right|\left|(\beta_{p}\left(\mathcal{C}\right))^{\perp}\right|=p^{n}$. In \cite{Dougherty3}, they construct ACD codes and use them to construct infinite families of binary LCD codes via the Gray map. We note that using the relations, ACD codes over $I_{2}$ and $I_{3}$ can be used to generate binary and ternary LCD codes. We also note that according to the first relation, there are many ACD codes over $I_{2}$, and this way of finding ACD codes is more flexible than those of finding ACD codes over $\mathbb{F}_{4}$ from binary codes in \cite{Kim2}. In \cite{Kim6}, they proved how LCD codes yield ACD codes relative to inner products considered in their work. We, however, use binary and ternary LCD codes that can be found from the already known classification of LCD codes to classify ACD codes over $I_{2}$ and $I_{3}$ using the relations.

Our paper is organized as follows. In Section~\ref{sec: preliminaries}, we review some fundamental concepts
of coding theory and rings. In Section~\ref{sec: ACD codes over I2}, we discuss ACD codes over $I_{2}$. In Section~\ref{sec: ACD codes over I3}, we investigate ACD codes over $I_{3}$. Finally, we conclude this article in Section~\ref{sec: conclusion}.

\section{Preliminaries}
\label{sec: preliminaries}

Here we briefly review some basics of the theory of error-correcting codes \cite{Pless1},\cite{Macwilliams1},\cite{Pless2}, \cite{Pretzel} and commutative rings of orders 4 and 9 \cite{Sole6},\cite{Sole2} that are needed to follow the subsequent material of this paper.

\indent
The {\em Hamming weight} of $\textbf{x}$ $\in$ $\mathbb{F}_{2}^{n}$ is the number of non-zero components, and is denoted by $\operatorname{wt}(\textbf{x})$. If $\mathcal{C}$ is a binary linear code, the {\em dual} of $\mathcal{C}$ is written as $\mathcal{C}^{\perp}$, where $\mathcal{C}^{\perp}=\left\{ \textbf{u}\in\mathbb{F}_{2}^{n}\mid \left\langle\mathbf{u},\mathbf{w}\right\rangle=0 \:\textrm{for all}\:\textbf{w}\in \mathcal{C}\right\}$  and $\left\langle \;, \; \right\rangle $ is the standard inner product. A code $\mathcal{C}$ is {\em self-orthogonal} if $\mathcal{C}\subseteq \mathcal{C}^{\perp}$.

\indent
We will study codes defined over the following ring $I_{2}$, classified as such in \cite{Fine}. The ring $I_{2}$ is defined by two generators $a$ and $b$ with the relations
\[I_{2}=\left\langle a,b\mid 2a=2b=0,\:a^{2}=b,\:ab=0\right\rangle. \]
Thus, $I_{2}$ is a ring with characteristic 2 and is composed of 4 elements $\left\{ 0,a,b,c\right\} $ where $c=a+b$. The addition and multiplication tables for $I_2$ are given below.

\begin{table}[ht]
\caption{Addition table of the ring $I_{2}$}
\label{tbl: add tbl I2}
$$\begin{tabular}{m{0.5cm}| m{1cm} m{1cm} m{1cm} m{0.4cm}}
\hline
+ & 0 & $a$ & $b$ & $c$\tabularnewline
\hline
0 & 0 & $a$ & $b$ & $c$\tabularnewline

$a$ & $a$ & 0 & $c$ & $b$\tabularnewline

$b$ & $b$ & $c$ & 0 & $a$\tabularnewline

$c$ & $c$ & $b$ & $a$ & 0\tabularnewline
\hline
\end{tabular}$$
\end{table}

\begin{table}[ht]
\caption{Multiplication table of the ring $I_{2}$}
\label{tbl: mult tbl I2}
$$\begin{tabular}{m{0.5cm} | m{1cm} m{1cm} m{1cm} m{0.4cm}}
\hline
$\times$ & 0 & $a$ & $b$ & $c$\tabularnewline
\hline
0 & 0 & 0 & 0 & 0\tabularnewline

$a$ & 0 & $b$ & 0 & $b$\tabularnewline

$b$ & 0 & 0 & 0 & 0\tabularnewline

$c$ & 0 & $b$ & 0 & $b$\tabularnewline
\hline
\end{tabular}$$
\end{table}

\indent
We note that the additive group of $I_{2}$ is the Klein-4 group.
From the multiplication table, the ring is commutative, and there is no multiplicative identity.
It is a local ring with maximal ideal $J=\left\{ 0,b\right\}$ and residue field $\left\{ 0,1\right\}$.
We observe that the ring $I_{2}$ may be included as an ideal $\left(u\right)$ in the finite unital
chain ring $R=\mathbb{F}_{2}\left[u\right]/\left(u^{3}\right)$. That is,
the ideal $\left(u\right)$ consists of  $xu + yu^{2}$ where $x,y\in\mathbb{F}_{2}$. The correspondence between $I_{2}$ and the ideal
$\left(u\right)$ is given by $0 \leftrightarrow 0$, $a \leftrightarrow u$,
$b \leftrightarrow u^{2}$, and $c \leftrightarrow u + u^{2}$. The roles of $a$ and $c$ could be reversed.
We denote the map of reduction modulo $J$ by $\alpha \colon I_{2}\rightarrow I_{2}/J\simeq\mathbb{F}_{2}$. Then we have $\alpha\left(0\right)=\alpha\left(b\right)=0$, and $\alpha\left(a\right)=\alpha\left(c\right)=1$. This map is extended naturally in a map from $I_{2}^{n}$ to $\mathbb{F}_{2}^{n}$.
There is another additive map that interacts in an interesting way given by $i \colon \mathbb{F}_{2}\rightarrow I_{2}$ where $i(x) = xb$. This is an injective homomorphism of abelian groups. We observe that $xy =(\alpha(x)\alpha(y))b$ for $x, y \in I_{2}$.
We denote the standard inner product on $I_{2}$ as $\left(\;,\;\right)$. Then for $\mathbf{x},\mathbf{y}\in I_{2}^{n}$, $\left(\mathbf{x},\mathbf{y}\right)=i\left(\left\langle \alpha\left(\mathbf{x}\right),\alpha\left(\mathbf{y}\right)\right\rangle \right)=\left\langle \alpha\left(\mathbf{x}\right),\alpha\left(\mathbf{y}\right)\right\rangle b$. As a consequence, $\mathbf{y}\in \mathcal{C}^{\perp}$ if and only if $\alpha\left(\mathbf{y}\right)\in\left(\alpha \left( \mathcal{C}  \right)\right)^{\perp}$.
An annihilator for $I_{2}$ is an element $x\in I_{2}$ such that $xy=0$ for all $y\in I_{2}$. Thus, we see from Table \ref{tbl: mult tbl I2} that there are two annihilators for $I_{2}$, namely $0$ and $b$. We note that if $N\left(\mathbf{y}\right)$ denotes the number of non-annihilators of a vector $\mathbf{y}\in I_{2}^{n}$, then $\left(\mathbf{y},\mathbf{y}\right)=N\left(\mathbf{y}\right)b$.
A {\em linear $I_{2}$-code $\mathcal{C}$ of length $n$} is an $I_{2}$-submodule of $I_{2}^{n}$. It is described by the $I_{2}$-span of the rows of a generator matrix. The {\em weight} of $\mathbf{x}\in I_{2}^{n}$ is the number of coordinates of $\mathbf{x}$ which are not zero.
The \emph{minimum distance} of an $I_2$-code $\mathcal{C}$ is equal to the minimum weight among all nonzero codewords of $\mathcal{C}$.
Two $I_{2}$-codes are called {\em equivalent} if one can be obtained from the other by a coordinate permutation. These two definitions are the same as those of binary codes.
We note that the elements of $I_{2}$ can be written as $c_{ij}=ia+jb$ where $0\leq i,j<2$. We define a Gray map $\phi_{2} \colon I_{2}\rightarrow\mathbb{F}_{2}^{2}$ by $\phi_{2}\left(ia+jb\right)=\left(i,j\right)$. With this definition, adjacent values in $I_{2}$ only differ by a single bit over $\mathbb{F}_{2}^{2}$. This Gray map $\phi_{2} \colon I_{2}\rightarrow\mathbb{F}_{2}^{2}$ is extended naturally to a map $\Phi_{2} \colon I_{2}^{n}\rightarrow\mathbb{F}_{2}^{2n}$ by $\Phi_{2}\left(c_{1}, \dotsc ,c_{n}\right)=\left(\phi_{2}\left(c_{1}\right), \dotsc ,\phi_{2}\left(c_{n}\right)\right)$. That is, for $c_{k}=i_{k}a+j_{k}b$, $\Phi_{2}\left( \mathbf{c} \right)=\left(i_{1},j_{1},i_{2},j_{2}, \dotsc ,i_{n},j_{n}\right)\in\mathbb{F}_{2}^{2n}$. We also define the Lee weight of $x=ia+jb\in I_{2}$ by $w_{L}\left(x\right)=\textrm{min}\left(i,2-i\right)+\textrm{min}\left(j,2-j\right)$. Then the Gray map preserves the distances because $w_{L}\left(x\right)=w_{H}\left( \phi_{2} \left(x\right)\right)=\textrm{min}\left(i,2-i\right)+\textrm{min}\left(j,2-j\right)$.

The ring $I_{3}$ is defined by two generators $a$ and $b$ with the relations
\[I_{3}=\left\langle a,b\mid3a=3b=0,\:a^{2}=b,\:ab=0\right\rangle. \]
$I_{3}$ is commutative without multiplicative identity and consists of $3^{2}$ elements, which can be written as $c_{ij}=ia+jb$ where $0\leq i,j<3$ \cite{Sole6}. Then we can write $I_{3}$ as $I_{3}=\left\{ ax+by \mid x,y\in\mathbb{F}_{3}\right\} $.
Thus, $c_{00}=0$, $c_{10}=a$, $c_{01}=b$, $c_{11}=a+b$, $c_{20}=2a$, $c_{02}=2b$, $c_{21}=2a+b$, $c_{12}=a+2b$, $c_{22}=2a+2b$.
$I_{3}$ contains a unique maximal ideal $J_{3}=\left\{ jb:0\leq j<3\right\}$. The reduction map modulo $J_{3}$ is defined as $\beta \colon I_{3}\mapsto I_{3}/J_{3} \simeq \mathbb{F}_{3}$ by $\beta\left(c_{ij}\right)=i$ where $0\leq i<3$ \cite{Sole6}.
The addition and multiplication tables are given below. We define a Gray map $\phi_{3} \colon I_{3}\rightarrow\mathbb{F}_{3}^{2}$ by $\phi_3\left(ia+jb\right)=\left(i,j\right)$. With this definition, adjacent values in $I_{3}$ only differ in exactly one coordinate over $\mathbb{F}_{3}^{2}$. This Gray map $\phi_{3} \colon I_{3}\rightarrow\mathbb{F}_{3}^{2}$ is extended naturally to a map $\Phi_{3} \colon I_{3}^{n}\rightarrow\mathbb{F}_{3}^{2n}$ by $\Phi_{3}\left(c_{1}, \dotsc ,c_{n}\right)=\left(\phi_{3}\left(c_{1}\right), \dotsc ,\phi_{3}\left(c_{n}\right)\right)$. That is, for $c_{k}=i_{k}a+j_{k}b$, $\Phi_{3}\left( \mathbf{c} \right)=\left(i_{1},j_{1},i_{2},j_{2}, \dotsc ,i_{n},j_{n}\right)\in\mathbb{F}_{3}^{2n}$. We also define the Lee weight of $x=ia+jb\in I_{3}$ by $w_{L}\left(x\right)=\textrm{min}\left(i,3-i\right)+\textrm{min}\left(j,3-j\right)$. Then the Gray map preserves the distances because $w_{L}\left(x\right)=w_{H}\left( \phi_3 \left(x\right)\right)=\textrm{min}\left(i,3-i\right)+\textrm{min}\left(j,3-j\right)$.
A {\em linear $I_{3}$-code $\mathcal{C}$ of length $n$} is an $I_{3}$-submodule of $I_{3}^{n}$.

The ring $I_{p}$ can be described in a similar manner \cite{Sole6}.
The ring $I_{p}$ in the classification of \cite{Fine} is defined by two generators $a$ and $b$ with the relations
\[I_{p}=\left\langle a,b\mid pa=pb=0,\:a^{2}=b,\:ab=0\right\rangle. \]
$I_{p}$ is commutative without multiplicative identity and consists of $p^{2}$ elements, which can be written as $c_{ij}=ia+jb$ where $0\leq i,j<p$. Then we can write $I_{p}$ as $I_{p}=\left\{ ax+by \mid x,y\in\mathbb{F}_{p}\right\} $.
$I_{p}$ contains a unique maximal ideal $J_{p}=\left\{ jb \mathrel{:} 0\leq j<p\right\}$. The reduction map modulo $J_{p}$ is defined as $\beta_{p} \colon I_{p}\mapsto I_{p}/J_{p} \simeq \mathbb{F}_{p}$ by $\beta_{p}\left(c_{ij}\right)=i$ where $0\leq i<p$.
A {\em linear $I_{p}$-code $\mathcal{C}$ of length $n$} is an $I_{p}$-submodule of $I_{p}^{n}$.
Two $I_{p}$-codes $\mathcal{C}$ and $\mathcal{C}'$ are {\em monomially equivalent} if there is an $n\times n$ monomial matrix $M$ such that $\mathcal{C}'=\left\{ \mathbf{c} M \mathrel{:} \mathbf{c} \in\mathcal{C}\right\} $. We shall use monomial equivalence for the classification purposes throughout our paper. Note that monomial equivalence and permutation equivalence are the same for $I_{2}$.

We recall from~\cite[p. 319]{Macwilliams1} that for a nonlinear block code $\mathcal{C}$ over the finite field $\operatorname{GF}(q)$ of length $n$ and minimum Hamming distance $d$ satisfies the Singleton bound given by

\begin{equation} \label{eq-Singleton}
d \le n - \log_{q} | \mathcal{C} | +1.
\end{equation}

 We can replace $\operatorname{GF}(q)$  by $I_p$ so that if $\mathcal{C}$ is a block code over $I_p$ of length $n$ (not necessarily an $I_p$-code), the Hamming distance $d$ of $\mathcal{C}$ satisfies the Singleton bound
\begin{equation} \label{eq-Singleton-Ip}
d \le n - \log_{p^2} | \mathcal{C} | +1.
\end{equation}
A code obtaining the upper bound of  (\ref{eq-Singleton-Ip}) is called a {\em maximum distance separable (MDS) code}.

We recall that for an $[n,k,d]$ linear code $\mathcal{C}$ over the field $\mathbb{F}_{2}$ or $\mathbb{F}_{3}$, the Lee weight satisfies the following Singleton bound:
\begin{equation}
d_{L}\left( \mathcal{C} \right)\leq n-k+1  \label{eq:Singleton-Lee}
\end{equation}
We will call a code satisfying the bound in (\ref{eq:Singleton-Lee}) an optimal code.

\begin{table}[ht]
\caption{Addition table of the ring $I_{3}$}
\resizebox{\textwidth}{!}{%
$$\begin{tabular}{m{0.5cm}| m{1.2cm} m{1.2cm} m{1.2cm} m{1.2cm} m{1.2cm} m{1.2cm} m{1.2cm} m{1.2cm} m{1.2cm}}
\hline
+ & $c_{00}$ & $c_{10}$ & $c_{01}$ & $c_{11}$ & $c_{20}$ & $c_{02}$ & $c_{21}$ & $c_{12}$ & $c_{22}$\tabularnewline
\hline
$c_{00}$ & 0 & $a$ & $b$ & $a+b$ & 2a & $2b$ & $2a+b$ & $a+2b$ & $2a+2b$\tabularnewline

$c_{10}$ & $a$ & $2a$ & $a+b$ & $2a+b$ & 0 & $a+2b$ & $b$ & $2a+2b$ & $2b$\tabularnewline

$c_{01}$ & $b$ & $a+b$ & $2b$ & $a+2b$ & $2a+b$ & 0 & $2a+2b$ & $a$ & $2a$\tabularnewline

$c_{11}$ & $a+b$ & $2a+b$ & $a+2b$ & $2a+2b$ & $b$ & $a$ & $2b$ & $2a$ & 0\tabularnewline

$c_{20}$ & $2a$ & 0 & $2a+b$ & $b$ & $a$ & $2a+2b$ & $a+b$ & $2b$ & $a+2b$\tabularnewline

$c_{02}$ & $2b$ & $a+2b$ & 0 & $a$ & $2a+2b$ & $b$ & $2a$ & $a+b$ & $2a+b$\tabularnewline

$c_{21}$ & $2a+b$ & $b$ & $2a+2b$ & $2b$ & $a+b$ & $2a$ & $a+2b$ & 0 & $a$\tabularnewline

$c_{12}$ & $a+2b$ & $2a+2b$ & $a$ & $2a$ & $2b$ & $a+b$ & 0 & $2a+b$ & $b$\tabularnewline

$c_{22}$ & $2a+2b$ & $2b$ & $2a$ & 0 & $a+2b$ & $2a+b$ & a & $b$ & $a+b$\tabularnewline
\hline
\end{tabular}$$}
\end{table}

\begin{table}[ht]
\caption{Multiplication table of the ring $I_{3}$}
\resizebox{\textwidth}{!}{%
$$\begin{tabular}{m{0.5cm}| m{1.2cm} m{1.2cm} m{1.2cm} m{1.2cm} m{1.2cm} m{1.2cm} m{1.2cm} m{1.2cm} m{1.2cm}}
\hline
$\times$ & $c_{00}$ & $c_{10}$ & $c_{01}$ & $c_{11}$ & $c_{20}$ & $c_{02}$ & $c_{21}$ & $c_{12}$ & $c_{22}$\tabularnewline
\hline
$c_{00}$ & 0 & 0 & 0 & 0 & 0 & 0 & 0 & 0 & 0\tabularnewline

$c_{10}$ & 0 & $b$ & 0 & $b$ & $2b$ & 0 & $2b$ & $b$ & $2b$\tabularnewline

$c_{01}$ & 0 & 0 & 0 & 0 & 0 & 0 & 0 & 0 & 0\tabularnewline

$c_{11}$ & 0 & $b$ & 0 & $b$ & $2b$ & 0 & $2b$ & $b$ & $2b$\tabularnewline

$c_{20}$ & 0 & $2b$ & 0 & $2b$ & $b$ & 0 & $b$ & $2b$ & $b$\tabularnewline

$c_{02}$ & 0 & 0 & 0 & 0 & 0 & 0 & 0 & 0 & 0\tabularnewline

$c_{21}$ & 0 & $2b$ & 0 & $2b$ & $b$ & 0 & $b$ & $2b$ & $b$\tabularnewline

$c_{12}$ & 0 & $b$ & 0 & $b$ & $2b$ & 0 & $2b$ & $b$ & $2b$\tabularnewline

$c_{22}$ & 0 & $2b$ & 0 & $2b$ & $b$ & 0 & $b$ & $2b$ & $b$\tabularnewline
\hline
\end{tabular}$$}
\end{table}

\section{ACD codes over $\boldsymbol{I_{2}}$}
\label{sec: ACD codes over I2}

\noindent
Definitions such as left-LCD codes and left-ACD codes over $E$ as well as binary ACD codes can be found in \cite{Kim1}.
We start this section by making similar definitions for $I_{2}$-codes.
Let $\mathcal{C}$ be an $I_{2}$-code of length $n$.
The dual $\mathcal{C}^{\perp}$ of $\mathcal{C}$ is defined by
\begin{center}
$\mathcal{C}^{\perp}=\left\{ \mathbf{y}\in I_{2}^{n}\mid\forall\mathbf{x}\in \mathcal{C},\left(\mathbf{x},\mathbf{y}\right)=0\right\}. $
\end{center}

Call a linear $I_{2}$-code $\mathcal{C}$ \textit{nice} if $\left|\mathcal{C}\right|\left|\mathcal{C}^{\perp}\right|=4^{n}$.
Define a linear $I_{2}$-code $\mathcal{C}$ to be LCD if it is nice and $\mathcal{C}\cap \mathcal{C}^{\perp}=\left\{ \mathbf{0}\right\} $.
An additive $I_{2}$-code of length $n$ is an additive subgroup of $I_{2}^{n}$.
An additive $I_{2}$-code $\mathcal{C}$ is called nice if $\left|\mathcal{C}\right|\left|\mathcal{C}^{\perp}\right|=4^{n}$.

\begin{definition}
An additive $I_{2}$-code $\mathcal{C}$ is ACD if it is nice and $\mathcal{C}\cap \mathcal{C}^{\perp}=\left\{ \mathbf{0}\right\} $.
\end{definition}

We shall investigate LCD and ACD codes over $I_{2}$, just as research on LCD and ACD codes over $E$ has been carried out. We start with a few simple examples to derive characteristics, if any, of LCD codes over $I_{2}$.
\begin{example}  \label{ex: I2-code J}
Let $J=\left\{ 0,b\right\}$. Then $J^{\perp}$ consists of elements of $I_{2}$ with which the inner product of the elements of J is $0$. Thus, $J^{\perp} = I_2  = \left\{ 0,a,b,c\right\} $,
$\left|J\right|\left|J^{\perp}\right|\neq4^{1}$, and $J\cap J^{\perp}\neq\left\{ 0 \right\} $.
Hence, J is not LCD.
\end{example}

\begin{example}  \label{ex: I2-code R_2}
Let $R_{2}=\left\{ 00,aa,bb,cc\right\} $. Then
$R_{2}^{\perp}=\left\{ 00,aa,bb,cc,0b,b0,ac,ca\right\} $,
$\left|R_{2}\right|\left|R_{2}^{\perp}\right|\neq4^{2}$, and $R_{2}\cap R_{2}^{\perp}\neq\left\{ \mathbf{0}\right\} $.
Hence, $R_{2}$ is not LCD.
\end{example}

The following theorem confirms the observation suggested by \Cref{ex: I2-code J,ex: I2-code R_2} that LCD codes over $I_{2}$ do not exist.
\begin{theorem}
There is no LCD code over $I_{2}$.
\end{theorem}

\begin{proof}
We choose any non-zero codeword $\mathbf{x}$ in a linear code $\mathcal{C}$ over $I_{2}$ .
If all the components of $\mathbf{x}$ consist of zeros and $b$s, then $\mathbf{x}$ is also in the dual of $\mathcal{C}$ since $b$ is an annihilator in $I_{2}$.
Otherwise, every coordinate of $a \mathbf{x}$ belongs to $\{ 0, b\}$; hence $a\mathbf{x}$ is non-zero and is also in the dual of $\mathcal{C}$.  Both cases imply that $\mathcal{C} \cap \mathcal{C}^{\perp} \ne \left\{ \mathbf{0}\right\} $.
Therefore, $\mathcal{C}$ is not LCD.
\end{proof}

A recent paper by Alahmadi et al. \cite{Sole3} confirms that nontrivial LCD codes over $I_{2}$ do not exist. Because there is no LCD code over $I_{2}$, we naturally turn our attention to ACD codes.
First, we show the following relation between ACD codes over $I_{2}$ and binary LCD codes. We note that an additive code over $\mathbb{F}_{2}$ is linear, so that ACD is LCD over $\mathbb{F}_{2}$.

\begin{theorem}  \label{thm: ACD iff LCD I2-code}
An additive code $\mathcal{C}\subseteq I_{2}^{n}$ is ACD if and only if $\alpha\left(\mathcal{C}\right)\subseteq\mathbb{F}_{2}^{n}$ is LCD {\color{blue}and $\left|\alpha\left(\mathcal{C}\right)\right|=\left|\mathcal{C}\right|$}.
\end{theorem}

{\color{blue}
\begin{proof}
Assume that $\mathcal{C}\subseteq I_{2}^{n}$ is ACD, i.e.,
$\mathcal{C}\cap\mathcal{C}^{\perp}=\left\{\mathbf{0}\right\}$ and
$\left|\mathcal{C}\right|\left|\mathcal{C}^{\perp}\right|=4^{n}$.
We will show that
$\alpha\left(\mathcal{C}\right)\cap\left(\alpha\left(\mathcal{C}\right)\right)^{\perp}
=\left\{\mathbf{0}\right\}$,
$\left|\alpha\left(\mathcal{C}\right)\right|
\left|(\alpha\left(\mathcal{C}\right))^{\perp}\right|=2^{n}$
and
$\left|\alpha\left(\mathcal{C}\right)\right|
=\left|\mathcal{C}\right|$.
Since we assume
$\mathcal{C}\cap\mathcal{C}^{\perp}
=\left\{\mathbf{0}\right\}$,
no non-zero element of $\mathcal{C}$ is in $\mathcal{C}^{\perp}$.
Suppose $\mathbf{c}\in\mathcal{C}$ is non-zero.
We may assume that $\alpha\left(\mathbf{c}\right)$ is non-zero.
If $\alpha\left(\mathbf{c}\right)$ is zero, there is nothing to prove.
We shall show that
$\alpha\left(\mathbf{c}\right)
\notin
\alpha\left(\mathcal{C}\right)
\cap
\left(\alpha\left(\mathcal{C}\right)\right)^{\perp}$.
Since
$\alpha\left(\mathbf{c}\right)
\in
\alpha\left(\mathcal{C}\right)$,
we only need to show that
$\alpha\left(\mathbf{c}\right)
\notin
\alpha\left(\mathcal{C}\right)^{\perp}$.
If $\mathbf{c}$ has an odd number of non-annihilators, then
wt$\left(\alpha\left(\mathbf{c}\right)\right)$ is odd.
Thus,
$\langle \alpha\left(\mathbf{c}\right) , 
\alpha\left(\mathbf{c}\right) \rangle \neq0$.
Therefore,
$\alpha\left(\mathbf{c}\right)
\notin
\left(\alpha\left(\mathcal{C}\right)\right)^{\perp}$.
If $\mathbf{c}$ has an even number of non-annihilators, then there exists
non-zero
$\mathbf{c}_{1}\in\mathcal{C}$
such that
$(\mathbf{c} , \mathbf{c}_{1}) \neq0$
and
$\alpha\left(\mathbf{c}_{1}\right)$
is non-zero since
$\mathbf{c}$
is not in
$\mathcal{C}^{\perp}$
and
$\mathbf{c}_{1}$
is non-zero.
Since
$(\mathbf{c} , \mathbf{c}_{1}) \neq0$,
non-annihilators in
$\mathbf{c}$
overlap at odd places with those in
$\mathbf{c}_{1}$.
Then it follows that
$\langle\alpha\left(\mathbf{c}\right) ,
\alpha\left(\mathbf{c}_{1}\right) \rangle \neq 0$.
Thus, we have found a codeword
$\alpha\left(\mathbf{c}_{1}\right)
\in
\alpha\left(\mathcal{C}\right)$
such that
$\langle\alpha\left(\mathbf{c}\right) ,
\alpha\left(\mathbf{c}_{1}\right) \rangle
\neq0$.
Therefore,
$\alpha\left(\mathbf{c}\right)
\notin
\left(\alpha\left(\mathcal{C}\right)\right)^{\perp}$.
Hence,
$\alpha\left(\mathcal{C}\right)
\cap
\left(\alpha\left(\mathcal{C}\right)\right)^{\perp}
=
\left\{\mathbf{0}\right\}$.
 Since the additive group of $I_{2}$ is isomorphic to $(\mathbb{F}_{2})^{2}$, every additive subgroup of $I_{2}^{n}$ is naturally an $\mathbb{F}_{2}$-vector space. Hence, $|\mathcal{C}|=2^{k}$ for some integer $k$.
Since 
$\left|\mathcal{C}\right|\left|\mathcal{C}^{\perp}\right|=4^{n} = 2^{2n}$,
we have
$\left|\mathcal{C}^{\perp}\right|
=
2^{2n-k}$.
 Because $0$ corresponds to $0$ or $b$ and $1$ to $a$ or $c$,
each element of
$\alpha\left(\mathcal{C}\right)^{\perp}$
corresponds to
$2^{n}$
elements of
$\mathcal{C}^{\perp}$.
Thus,
$\left|(\alpha\left(\mathcal{C}\right))^{\perp}\right|
=
2^{n-k}$.
We note that for $\mathcal{C}$ to be an ACD code,
$\left|\mathcal{C}\right|$
has to be equal to
$\left|\alpha\left(\mathcal{C}\right)\right|$.
We will show that
$\left|\mathcal{C}\right|
=
\left|\alpha\left(\mathcal{C}\right)\right|$.
We note that since $\alpha$ is a reduction map,
$\left|\mathcal{C}\right|
\geq
\left|\alpha\left(\mathcal{C}\right)\right|$.
Suppose, for a contradiction, that
$\left|\mathcal{C}\right|
>
\left|\alpha\left(\mathcal{C}\right)\right|$.
Then for a certain
$\alpha\left(\mathbf{c}\right)
\in
\alpha\left(\mathcal{C}\right)$,
there exist at least two distinct codewords
$\mathbf{c},\mathbf{c}_{1}
\in
\mathcal{C}$
such that
$\alpha\left(\mathbf{c}\right)
=
\alpha\left(\mathbf{c}_{1}\right)$.
Thus,
$\alpha\left(\mathbf{c}\right)
+
\alpha\left(\mathbf{c}_{1}\right)
=
0$.
Then
$\mathbf{c}+\mathbf{c}_{1}$
is non-zero and  all of its coordinates belong to $\{ 0 , b\}$.
This non-zero
$\mathbf{c}+\mathbf{c}_{1}$
is then in both
$\mathcal{C}$
and
$\mathcal{C}^{\perp}$,
which is a contradiction.
Therefore,
$\left|\mathcal{C}\right|
=
\left|\alpha\left(\mathcal{C}\right)\right|$.
Thus,
$\left|\alpha\left(\mathcal{C}\right)\right|
\left|(\alpha\left(\mathcal{C}\right))^{\perp}\right|
=
2^{n}$.

For the other direction, assume that
$\alpha\left(\mathcal{C}\right)
\cap
\left(\alpha\left(\mathcal{C}\right)\right)^{\perp}
=
\left\{\mathbf{0}\right\}$,
$\left|\alpha\left(\mathcal{C}\right)\right|
\left|(\alpha\left(\mathcal{C}\right))^{\perp}\right|
=
2^{n}$
and
$\left|\alpha\left(\mathcal{C}\right)\right|
=
\left|\mathcal{C}\right|$.
We will show that
$\mathcal{C}\cap\mathcal{C}^{\perp}
=
\left\{\mathbf{0}\right\}$
and
$\left|\mathcal{C}\right|
\left|\mathcal{C}^{\perp}\right|
=
4^{n}$.
Since
$\alpha\left(\mathcal{C}\right)
\cap
\left(\alpha\left(\mathcal{C}\right)\right)^{\perp}
=
\left\{\mathbf{0}\right\}$,
no non-zero element of
$\alpha\left(\mathcal{C}\right)$
is in
$\left(\alpha\left(\mathcal{C}\right)\right)^{\perp}$.
Also, since
$\left|\alpha\left(\mathcal{C}\right)\right|
=
\left|\mathcal{C}\right|$,
there is only one non-zero element of
$\mathcal{C}$
that corresponds to each non-zero element of
$\alpha\left(\mathcal{C}\right)$.
Suppose
$\alpha\left(\mathbf{c}\right)
\in
\alpha\left(\mathcal{C}\right)$
is non-zero.
We shall show that
$\mathbf{c}
\notin
\mathcal{C}
\cap
\mathcal{C}^{\perp}$.
Since
$\mathbf{c}
\in
\mathcal{C}$,
we only need to show that
$\mathbf{c}
\notin
\mathcal{C}^{\perp}$.
If
$\alpha\left(\mathbf{c}\right)$
 contains an odd number of coordinates equal to $1$, then
$\mathbf{c}$
has an odd number of non-annihilators.
Thus,
$( \mathbf{c} , \mathbf{c} ) \neq0$.
Therefore,
$\mathbf{c}
\notin
\mathcal{C}^{\perp}$.
If
$\alpha\left(\mathbf{c}\right)$
 contains an even number of coordinates equal to $1$, then there exists
$\mathbf{c}_{1}\in\mathcal{C}$
such that
$\langle \alpha\left(\mathbf{c}\right)
,
\alpha\left(\mathbf{c}_{1}\right) \rangle
\neq0$.
Since
$\langle \alpha\left(\mathbf{c}\right)
,
\alpha\left(\mathbf{c}_{1}\right) \rangle
\neq0$,
 the coordinates equal to $1$ in $\alpha\left(\mathbf{c}\right)$ and $\alpha\left(\mathbf{c}_1 \right)$ overlap in an odd number of positions..
Then it follows that
$( \mathbf{c} , \mathbf{c}_{1} )
=
b
\neq0$.
Thus, we have found a codeword
$\mathbf{c}_{1}
\in
\mathcal{C}$
such that
$( \mathbf{c} , \mathbf{c}_{1} )
\neq0$.
Therefore,
$\mathbf{c}
\notin
\mathcal{C}^{\perp}$.
Thus,
$\mathcal{C}
\cap
\mathcal{C}^{\perp}
=
\left\{\mathbf{0}\right\}$.
Now, we have
$\left|\alpha\left(\mathcal{C}\right)\right|
=
\left|\mathcal{C}\right|
=
2^{k}$
for some $k$.
Thus,
$\left|(\alpha\left(\mathcal{C}\right))^{\perp}\right|
=
2^{n-k}$.
As observed earlier,
because $0$ corresponds to $0$ or $b$ and $1$ to $a$ or $c$, each element of
$\alpha\left(\mathcal{C}\right)^{\perp}$
corresponds to
$2^{n}$
elements of
$\mathcal{C}^{\perp}$.
Thus,
$\left|\mathcal{C}^{\perp}\right|
=
2^{2n-k}$.
Therefore,
$\left|\mathcal{C}\right|
\left|\mathcal{C}^{\perp}\right|
=
4^{n}$.
\end{proof}
}

We can tell from \Cref{thm: ACD iff LCD I2-code} that there exists an ACD code of length $n$ with cardinality $2^{k}$. This ACD code is the preimage of a binary $[n,k,d]$ LCD code under the map $\alpha$. In the following example, we show that a simple code of length 1 is ACD using \Cref{thm: ACD iff LCD I2-code}. Note that throughout this section, we only deal with codes $\mathcal{C}$ that satisfy the condition $\left|\alpha\left(\mathcal{C}\right)\right|=\left|\mathcal{C}\right|$.

\begin{example}
For $\mathcal{C}=\left\{ 0,a\right\} $,
we note that
$\alpha\left(\mathcal{C}\right)=\left\{ 0,1\right\}$ and
$\alpha\left(\mathcal{C}\right)^{\perp}=\left\{ 0\right\} $.
Then we have
$\left|\alpha\left(\mathcal{C}\right)\right|\left|(\alpha\left(\mathcal{C}\right))^{\perp}\right|=2^{1}$ and
$\alpha\left(\mathcal{C}\right) \cap \alpha\left(\mathcal{C}\right)^{\perp}=\left\{ \mathbf{0}\right\} $.
Thus, $\alpha\left(\mathcal{C}\right)$ is a binary LCD code. By \Cref{thm: ACD iff LCD I2-code}, $\mathcal{C}$ is an ACD code.
\end{example}

From now on, we classify ACD codes for $n=1, 2, 3$ and partially for $n=4, 5$ using \Cref{thm: ACD iff LCD I2-code} together with the classification of binary LCD codes \cite{Harada}.
The next theorem classifies ACD codes for $n=1$.
\begin{theorem}
For $n=1$, there are only three ACD codes over $I_{2}$.
\end{theorem}

\begin{proof}
We classify ACD codes of length 1 using \Cref{thm: ACD iff LCD I2-code} as follows. There are only two binary LCD codes of length 1. One is $\left\{ 0\right\} $. The corresponding ACD code over $I_{2}$ is $\left\{ 0\right\} $. The other binary LCD code is $\alpha\left(\mathcal{C}\right)=\left\{ 0,1\right\} $. There are four codes of length 1 over $I_{2}$ that correspond to $\alpha\left(\mathcal{C}\right)$, namely $\left\{ 0,a\right\} $, $\left\{ 0,c\right\} $, $\left\{ b,a\right\} $ and $\left\{ b,c\right\} $. Among these, $\left\{ 0,a\right\} $ and $\left\{ 0,c\right\} $ are additive and ACD codes.
\end{proof}

In the following example, we show that a simple code of length 2 is ACD using \Cref{thm: ACD iff LCD I2-code}.
\begin{example}
For $\mathcal{C}=\left\{ 00, a0\right\} $,
we note that
$\alpha\left(\mathcal{C}\right)=\left\{ 00, 10\right\}$ and
$\alpha\left(\mathcal{C}\right)^{\perp}=\left\{ 00, 01\right\} $.
Then we have
$\left|\alpha\left(\mathcal{C}\right)\right|\left|(\alpha\left(\mathcal{C}\right))^{\perp}\right|=2^{2}$ and
$\alpha\left(\mathcal{C}\right) \cap \alpha\left(\mathcal{C}\right)^{\perp}=\left\{ \mathbf{0}\right\} $.
Hence, $\alpha\left(\mathcal{C}\right)$ is a binary LCD code. By \Cref{thm: ACD iff LCD I2-code}, $\mathcal{C}$ is an ACD code.
\end{example}

The next theorem classifies ACD codes for $n=2$.
\begin{theorem}
For $n=2$, there are exactly 15 ACD codes up to equivalence over $I_{2}$.
\end{theorem}

\begin{proof}
We divide the proof into three parts depending on the cardinality of a code.
\begin{enumerate}

\item[{(i)}] cardinality one
\\
$\left\{ 00\right\} $ is a trivial ACD code.
\item[{(ii)}] cardinality two
\\
According to \cite{Harada}, there is only one binary $[2,1]$ LCD code up to equivalence, which is $\alpha\left(\mathcal{C}\right)=\left\{ 00, 10\right\} $. By \Cref{thm: ACD iff LCD I2-code}, ACD codes $\mathcal{C}$ that correspond to this code are exactly $\left\{ 00,a0\right\} $, $\left\{ 00,c0\right\} $, $\left\{ 00,ab\right\} $, $\left\{ 00,cb\right\} $.

\item[{(iii)}] cardinality four
\\
The next possible cardinality of an additive code is four because the addition of two non-zero codewords results in a third non-zero codeword. According to \cite{Harada},  $\alpha\left(\mathcal{C}\right)=\left\{ 00, 10, 01, 11\right\} $ is also a trivial binary $[2,2]$ LCD code. To count the number of ACD codes that correspond to this code using \Cref{thm: ACD iff LCD I2-code}, we note that $a0$ corresponds to 10 and there are four choices that correspond to 01, namely $0a$, $0c$, $ba$ and $bc$. Also, $ab$ also corresponds to 10, but we have only three choices for 01 because $ab$ and $0a$ are the same as $ba$ and $a0$ when permuted. Likewise, we have two codes when we correspond $c0$ to 10, and one code when we correspond $cb$ to 10. Therefore, we have altogether ten ACD codes of this type.
\end{enumerate}
This completes the proof.
\end{proof}

There is a trivial ACD code with cardinality one for $n=3$, namely $\left\{ 000\right\} $.
The next theorem counts ACD codes with cardinality two for $n=3$.
\begin{theorem}
For $n=3$, there are exactly 10 ACD codes with cardinality two up to equivalence.
\end{theorem}

\begin{proof}
According to \cite{Harada}, there are only two binary $[3,1]$ LCD codes up to equivalence, one of which is $\alpha\left(\mathcal{C}\right)=\left\{ 000, 111\right\} $. By \Cref{thm: ACD iff LCD I2-code}, there are at most $2^{3}=8$ ACD codes over $I_{2}$ corresponding to this binary LCD code. From this eight we take away four because $acc$ is the same as $cac$ and $cca$ when permuted, and $aac$ is the same as $aca$ and $caa$ when permuted.
According to \cite{Harada}, there is another inequivalent binary $[3,1]$ LCD code $\alpha\left(\mathcal{C}\right)=\left\{ 000, 100\right\} $. By \Cref{thm: ACD iff LCD I2-code}, there are $2^{3}-2=6$ ACD codes over $I_{2}$ corresponding to this binary LCD code. Here, we take away two because $b0$ and $0b$ are the same when permuted.
\end{proof}

The next theorem counts ACD codes with cardinality four for $n=3$. In the proof, we will briefly see the merit of \Cref{thm: ACD iff LCD I2-code} that with one binary LCD code, the number of ACD codes we can find over $I_{2}$ is amplified and upper bounded to a few dozens.
\begin{theorem}
For $n=3$, there are exactly 52 ACD codes with cardinality four up to equivalence.
\end{theorem}

\begin{proof}
We may use \Cref{thm: ACD iff LCD I2-code} for finding these ACD codes, but it requires an exhaustive search because there are no apparent patterns among these ACD codes. We illustrate this method with a case study below. According to \cite{Harada}, there are only two binary $[3,2]$ LCD codes up to equivalence, one of which is $\alpha\left(\mathcal{C}\right)=\left\{ 000, 110, 101, 011\right\} $. We note that 110 and 101 are the generators, so there are at most $2^{6}$ ACD codes correspoding to this binary LCD code. For classification, choose $aa0$ for 110, then there are 4 vectors that start with $a$ for 101, namely, $a0a$, $a0c$, $aba$ and $abc$, which result in 4 inequivalent ACD codes. There are also 4 vectors that start with $c$ for 101, namely $cba$, $cbc$, $c0a$, $c0c$, the first two of which result in inequivalent ACD codes. But for the remaining 2 vectors, we have codes $\left\{ 000,aa0,c0a,baa\right\} $ and $\left\{ 000,aa0,c0c,bac\right\} $. But these two codes are equivalent to previously obtained codes $\left\{ 000,aa0,aba,0ca\right\} $ and $\left\{ 000,aa0,abc,0cc\right\} $, respectively. One may continue this exhaustive search to show that there are 16 ACD codes.
The other binary $[3,2]$ LCD code up to equivalence is $\alpha\left(\mathcal{C}\right)=\left\{ 000, 100, 010, 110\right\} $. We may use Theorem 3.5 for finding 36 ACD codes corresponding to the LCD code in a similar manner.
\end{proof}

The next theorem gives the number of ACD codes with cardinality eight for $n=3$.
\begin{theorem}
For $n=3$, there are exactly 104 ACD codes with cardinality eight up to equivalence.
\end{theorem}

\begin{proof}
According to a Magma computation with exhaustive search, there are exactly 104 inequivalent ACD codes of this type.
\end{proof}

Next, we give the number of ACD codes when $n=4$ with cardinality two.
\begin{theorem}
For $n=4$, there are exactly 16 ACD codes over $I_{2}$ with cardinality two up to equivalence.
\end{theorem}

\begin{proof}
According to \cite{Harada}, there are only two binary $[4,1]$ LCD codes up to equivalence. One is $\alpha\left(\mathcal{C}\right)=\left\{ 0000, 1000\right\} $ and the other is $\alpha\left(\mathcal{C}\right)=\left\{ 0000, 1110\right\} $. By \Cref{thm: ACD iff LCD I2-code}, there are at most $2^{4}$ ACD codes over $I_{2}$ corresponding to the first one, and another $2^{4}$ ACD codes corresponding to the second one. Among the $2^{4}$ ACD codes corresponding to the first one, $ab00$ is equivalent to $a0b0$ and $a00b$. Likewise for $cb00$, $abb0$, and $cbb0$. Thus, there are eight inequivalent codes, namely, $a000$, $c000$, $ab00$, $cb00$, $abb0$, $cbb0$, $abbb$, $cbbb$, each combined with $0000$. Similarly, there are eight inequivalent ACD codes corresponding to the second binary LCD code. Therefore, there are 16 inequivalent ACD codes over $I_{2}$ with cardinality two.
\end{proof}

The next theorem gives the number of ACD codes when $n=4$ with cardinality four which is obtained with the help of Magma computation.
\begin{theorem}
For $n=4$, there are exactly 308 ACD codes over $I_{2}$ with cardinality four.
\end{theorem}

\begin{proof}
According to \cite{Harada}, there are four binary $[4,2]$ LCD codes up to equivalence. One is $\alpha\left(\mathcal{C}\right)=\left\{ 0000, 1000, 0100, 1100\right\} $. By \Cref{thm: ACD iff LCD I2-code} and the fact that 1000 and 0100 are generators, there are at most $2^{8}$ ACD codes over $I_{2}$ corresponding to this binary LCD code. The other three binary $[4,2]$ LCD codes are $\left\{ 0000, 1000, 0111, 1111\right\} $, $\left\{ 0000, 1100, 1010, 0110\right\} $, $\left\{ 0000, 1100, 1011, 0111\right\} $. By a similar reasoning as before, there are at most $2^{8}$ ACD codes over $I_{2}$ corresponding to each of these binary LCD codes. Therefore, there are altogether at most $4\cdot2^{8}$ ACD codes over $I_{2}$ with cardinality four.
An exhaustive Magma search shows that there are exactly 308 ACD codes over $I_{2}$ with cardinality four.
\end{proof}

The next theorem gives the number of ACD codes with cardinality eight for $n=4$.
\begin{theorem}
For $n=4$, there are exactly 1504 ACD codes with cardinality eight up to equivalence.
\end{theorem}

\begin{proof}
An exhaustive Magma computation shows that there are 1504 inequivalent ACD codes of this type.
\end{proof}

Next, we give the number of ACD codes when $n=5$ with cardinality two.
\begin{theorem}
For $n=5$, there are exactly 28 ACD codes over $I_{2}$ with cardinality two up to equivalence.
\end{theorem}

\begin{proof}
According to \cite{Harada}, there are only three binary $[5,1]$ LCD codes up to equivalence. One is $\alpha\left(\mathcal{C}\right)=\left\{ 00000, 10000\right\} $, another is $\alpha\left(\mathcal{C}\right)=\left\{ 00000, 11100\right\} $,  and the other is $\alpha\left(\mathcal{C}\right)=\left\{ 00000, 11111\right\} $. By \Cref{thm: ACD iff LCD I2-code}, there are at most $2^{5}$ ACD codes over $I_{2}$ corresponding to the first one, and another $2^{5}$ ACD codes corresponding to the second one, and another $2^{5}$ ACD codes corresponding to the third one. Therefore there are at most 96 ACD codes over $I_{2}$ with cardinality two. Now, among the $2^{5}$ ACD codes corresponding to the first LCD code, there are 10 inequivalent ACD codes. We have two choices, $a$ or $c$, for the first symbol. For each choice for the first symbol, there are five options for the remaining four symbols that depend on how many $b$s they contain. Therefore, there are 10 inequivalent ACD codes. By similar reasoning, there are twelve and six inequivalent ACD codes corresponding to the second and third LCD code respectively. Therefore, we have a total of 28 inequivalent ACD codes with cardinality two.
\end{proof}

The classification of ACD codes over $I_{2}$ is summarized in Table 5. For each length and cardinality, the number of ACD codes is given in the third column. The highest Hamming and Lee minimum distances are given in the next two columns. The last column provides examples of ACD codes with the highest Lee and hence Hamming minimum distances. MDS over $I_{2}$ is marked with $\dagger$ and binary optimal LCD code is marked with $^*$. We note that for all length and cardinality, we have MDS over $I_{2}$, and in most cases we have binary optimal LCD code.

\begin{table}[h]
\begin{center}
\caption{Summary of classifications over $I_{2}$. Here $\dagger$ denotes MDS over $I_2$ and * denotes a binary optimal LCD code}\label{tab1}%
\resizebox{\textwidth}{!}{%
\begin{tabular}{p{0.3cm}>{\centering}p{2.0cm}>{\centering}p{2.0cm}>{\centering}p{2.0cm}>{\centering}p{2.0cm}>{\centering}p{4.0cm}}
\hline
$n$ & Cardinality  &  \#ACD codes & Highest $d_H$($d_L$)  & Binary LCD & Example \tabularnewline
\hline
1 & 2 & 2 & $1^{\dagger}$($1$) & $\left[1,1,1\right]^{*}$ & $\left\{ 0,a\right\} $ \tabularnewline

2 & 2 & 4 & $2^{\dagger}$(2) & $\left[2,1,1\right]^{*}$ & $\left\{ 00,ab\right\} $ \tabularnewline

2 & 4 & 10 & $2^{\dagger}$($2$) & $\left[2,2,1\right]^{*}$ & $\left\{ 00,cb,aa,bc\right\} $      \tabularnewline

3 & 2 & 10 & $3^{\dagger}$(3) & $\left[3,1,3\right]^{*}$ & $\left\{ 000,aaa\right\} $       \tabularnewline

3 & 4 & 52 & $3^{\dagger}$($3$) & $\left[3,2,2\right]^{*}$ & $\left\{ 000,aab,bcc,cba\right\} $                           \tabularnewline

3 & 8 & 104 & $2^{\dagger}$($2$) & $\left[3,3,1\right]^{*}$ & $\{ 000,a0a,bbc,cbb,bcb,$\\
                                                                            $aa0,ccc,0aa\}$   \tabularnewline
4 & 2 & 16 & $4^{\dagger}$(4) & $\left[4,1,1\right]$ & $\left\{ 0000,abbb\right\} $\tabularnewline

4 & 4 & 308 & $4^{\dagger}$(4)  & $\left[4,2,1\right]$ & $\left\{ 0000,abbb,ccca,baac\right\} $    \tabularnewline

4 & 8 & 1504 & $3^{\dagger}$($3$) & $\left[4,3,1\right]^{*}$ & $\{0000,bbcb,ba0b,ccba,$\\
                                                                            $0ccc,abbc,c0ab,aaa0\}$   \tabularnewline

5 & 2 & 28 & $5^{\dagger}$(5) & $\left[5,1,1\right]$ & $\left\{ 00000,abbbb\right\} $ \tabularnewline
\hline
\end{tabular}}

\end{center}
\end{table}

\section{ACD codes over $\boldsymbol{I_{3}}$}
\label{sec: ACD codes over I3}
The concept of nice codes was first introduced in \cite{Sole7}. They defined nice codes because for the ring of 4 elements they considered, it is not necessarily the case that $\left|\mathcal{C}\right|\left|\mathcal{C}^{\perp}\right|=4^{n}$. The following example shows that we also need the concept of nice codes for the ring $I_{3}$.
\begin{example}
Let $\mathcal{C}=I_{3}=\left\{c_{00}, c_{10}, c_{01}, c_{11}, c_{20}, c_{02}, c_{21}, c_{12}, c_{22} \right\}$. Then $\mathcal{C}^{\perp}=\left\{c_{00}, c_{01}, c_{02}\right\} $.
So, $\left|\mathcal{C}\right|\left|\mathcal{C}^{\perp}\right|=9\cdot3\neq9$.
\end{example}

Call a linear $I_{3}$-code $\mathcal{C}$ nice if $\left|\mathcal{C}\right|\left|\mathcal{C}^{\perp}\right|=9^{n}$.
Define a linear $I_{3}$-code $\mathcal{C}$ to be LCD if it is nice and $\mathcal{C}\cap \mathcal{C}^{\perp}=\left\{ \mathbf{0}\right\} $.
An additive $I_{3}$-code of length $n$ is an additive subgroup of $I_{3}^{n}$.
An additive $I_{3}$-code $\mathcal{C}$ is called nice if $\left|\mathcal{C}\right|\left|\mathcal{C}^{\perp}\right|=9^{n}$.

\begin{definition}
An additive $I_{3}$-code $\mathcal{C}$ is ACD if it is nice and $\mathcal{C}\cap \mathcal{C}^{\perp}=\left\{ \mathbf{0}\right\} $.
\end{definition}

We shall investigate LCD and ACD codes over $I_{3}$, just as we carried out the same kind of research over $I_{2}$. We start with a few simple examples to see if there are simple LCD codes over $I_{3}$.
\begin{example}
Let $J_{3}=\left\{ 0,b,2b\right\}=\left\{ c_{00},c_{01},c_{02}\right\}$. Then
\begin{equation*}
J^{\perp}=I_{3}=\left\{ c_{00}, c_{10}, c_{01}, c_{11}, c_{20}, c_{02}, c_{21}, c_{12}, c_{22}\right\},
\end{equation*}
$\left|J_{3}\right|\left|J_{3}^{\perp}\right|=3\cdot9\neq9^{1}$ and $J_{3}\cap J_{3}^{\perp}\neq\left\{ 0 \right\} $.
Hence, $J_{3}$ is not LCD.
\end{example}

\begin{example}
Let $R_{3}=\left\{ c_{00}c_{00}, c_{10}c_{10}, c_{01}c_{01}, c_{11}c_{11}, c_{20}c_{20}, c_{02}c_{02}, c_{21}c_{21}, c_{12}c_{12}, c_{22}c_{22}\right\} $. Then
\begin{align*}
R_{3}^{\perp}=&\left\{ c_{00}c_{00}, c_{01}c_{01}, c_{02}c_{02}, c_{00}c_{01}, c_{01}c_{00}, c_{00}c_{02}, c_{02}c_{00}, c_{01}c_{02}, c_{02}c_{01}\right\}.
\end{align*}
$\left|R_{3}\right|\left|R_{3}^{\perp}\right|=9\cdot9=9^{2}$ but $R_{3}\cap R_{3}^{\perp}\neq\left\{ \mathbf{0}\right\} $.
Hence, $R_{3}$ is not LCD.
\end{example}

The following theorem tells us that, just as in the case of $I_{2}$, it is not possible to find an LCD code over $I_{3}$.
\begin{theorem}
There is no LCD code over $I_{3}$.
\end{theorem}

\begin{proof}
Choose any non-zero codeword $\mathbf{x}$ in a linear code $\mathcal{C}$ over $I_{3}$ .
If all the components of $\mathbf{x}$ are in $J_{3}=\left\{ 0,b,2b\right\} $, then $\mathbf{x}$ is also in the dual of $\mathcal{C}$ since the elements of $J_{3}$ are annihilators in $I_{3}$.
Otherwise, all the components of $a\mathbf{x}$ in $\mathcal{C}$ consist of elements of $J_{3}$, hence $a\mathbf{x}$ is non-zero and is also in the dual of $\mathcal{C}$.  Both cases imply that $\mathcal{C} \cap \mathcal{C}^{\perp} \ne \left\{ \mathbf{0}\right\} $.
Therefore $\mathcal{C}$ is not LCD.
\end{proof}

The above theorem can be generalized to any ring $I_{p}$. It can be proved with a similar reasoning.
\begin{theorem}
There is no LCD code over $I_{p}$.
\end{theorem}

Because there is no LCD code over $I_{3}$, we naturally turn our attention to ACD codes.
We show the following relation, which is an extension of Theorem 3.5 on ring $I_{2}$ to ring $I_{3}$, between ACD codes over $I_{3}$ and ternary LCD codes.
\begin{theorem} \label{thm: ACD iff LCD I3-code}
An additive code $\mathcal{C}\subseteq I_{3}^{n}$ is ACD if and only if $\beta\left(\mathcal{C}\right)\subseteq\mathbb{F}_{3}^{n}$ is LCD {\color{blue}and $\left|\beta\left(\mathcal{C}\right)\right|=\left|\mathcal{C}\right|$}.
\end{theorem}

{\color{blue}
\begin{proof}
Assume that $\mathcal{C}\subseteq I_{3}^{n}$ is ACD, i.e., $\mathcal{C}\cap\mathcal{C}^{\perp}=\left\{ \mathbf{0}\right\} $ and
$\left|\mathcal{C}\right|\left|\mathcal{C}^{\perp}\right|=9^{n}$. We will show that $\beta\left(\mathcal{C}\right)\cap\left(\beta\left(\mathcal{C}\right)\right)^{\perp}=\left\{ \mathbf{0}\right\}$, $\left|\beta\left(\mathcal{C}\right)\right|\left|(\beta\left(\mathcal{C}\right))^{\perp}\right|=3^{n}$ and $\left|\beta\left(\mathcal{C}\right)\right|=\left|\mathcal{C}\right|$. Since we assume $\mathcal{C}\cap\mathcal{C}^{\perp}=\left\{ \mathbf{0}\right\}$, no non-zero element of $\mathcal{C}$ is in $\mathcal{C}^{\perp}$.
Suppose $\mathbf{c}_{1}=c_{i_{1}j_{1}}\in\mathcal{C}$ is non-zero. We may assume that $\beta\left(\mathbf{c}_{1}\right)$ is non-zero. If $\beta\left(\mathbf{c}_{1}\right)$ is zero, there is nothing to prove. We shall show that $\beta\left(\mathbf{c}_{1}\right)\notin\beta\left(\mathcal{C}\right)\cap\left(\beta\left(\mathcal{C}\right)\right)^{\perp}$. Since $\beta\left(\mathbf{c}_{1}\right)\in\beta\left(\mathcal{C}\right)$, we only need to show that $\beta\left(\mathbf{c}_{1}\right)\notin\beta\left(\mathcal{C}\right)^{\perp}$.
Since $\mathbf{c}_{1}=c_{i_{1}j_{1}}$ is non-zero and $\mathbf{c}_{1}\notin\mathcal{C}^{\perp}$, there is a non-zero $\mathbf{c}_{2}=c_{i_{2}j_{2}}\in\mathcal{C}$ such that $( \mathbf{c}_{1} , \mathbf{c}_{2} ) \neq0$.
Since $( \mathbf{c}_{1} , \mathbf{c}_{2} ) \neq0$, we see that both $i_{1}$ and $i_{2}$ are non-zero. Then it follows that $\beta\left(\mathbf{c}_{2}\right)$ is non-zero and $\langle \beta\left(\mathbf{c}_{1}\right) , \beta\left(\mathbf{c}_{2}\right) \rangle \neq0$. Thus, we have found a codeword $\beta\left(\mathbf{c}_{2}\right)\in\beta\left(\mathcal{C}\right)$ such that $\langle \beta\left(\mathbf{c}_{1}\right) , \beta\left(\mathbf{c}_{2}\right) \rangle \neq0$. Therefore, $\beta\left(\mathbf{c}_{1}\right)\notin\left(\beta\left(\mathcal{C}\right)\right)^{\perp}$. Hence, $\beta\left(\mathcal{C}\right)\cap\left(\beta\left(\mathcal{C}\right)\right)^{\perp}=\left\{ \mathbf{0}\right\}$.
 Since the additive group of $I_{3}$ is isomorphic to $(\mathbb{F}_{3})^{2}$, every additive subgroup of $I_{3}^{n}$ is naturally an $\mathbb{F}_{3}$-vector space. Hence, $|\mathcal{C}|=3^{k}$ for some integer $k$. Since $|\mathcal{C}||\mathcal{C}^{\perp}|=9^{n} = 3^{2n}$, it follows that 
$|\mathcal{C}^{\perp}|
=
3^{2n-k}$.
Because 0 corresponds to 0 or $b$ or $2b$ and 1 to $a$ or $a+b$ or $a+2b$ and 2 to $2a$ or $2a+b$ or $2a+2b$, each element of $\beta\left(\mathcal{C}\right)^{\perp}$ corresponds to $3^{n}$ elements of $\mathcal{C}^{\perp}$.
Thus, $\left|(\beta\left(\mathcal{C}\right))^{\perp}\right|=3^{n-k}$.
We note that for $\mathcal{C}$ to be an ACD code, $\left|\mathcal{C}\right|$ has to be equal to $\left|\beta\left(\mathcal{C}\right)\right|$.
We will show that $\left|\mathcal{C}\right|=\left|\beta\left(\mathcal{C}\right)\right|$. We note that since $\beta$ is a reduction map, $\left|\mathcal{C}\right|\geq\left|\beta\left(\mathcal{C}\right)\right|$. Suppose, for a contradiction, that $\left|\mathcal{C}\right|>\left|\beta\left(\mathcal{C}\right)\right|$. Then for a certain $\beta\left(\mathbf{c}_{1}\right)\in\beta\left(\mathcal{C}\right)$, there exist at least two distinct codewords $\mathbf{c}_{1},\mathbf{c}_{2}\in\mathcal{C}$ such that $\beta\left(\mathbf{c}_{1}\right)=\beta\left(\mathbf{c}_{2}\right)$. Thus, $2\beta\left(\mathbf{c}_{1}\right)+\beta\left(\mathbf{c}_{2}\right)=0$. Then $2\mathbf{c}_{1}+\mathbf{c}_{2}$ is non-zero and  all of its coordinates belong to $\{ 0 , b , 2b \}$. This non-zero $2\mathbf{c}_{1}+\mathbf{c}_{2}$ is then in both $\mathcal{C}$ and $\mathcal{C}^{\perp}$, which is a contradiction. Therefore, $\left|\mathcal{C}\right|=\left|\beta\left(\mathcal{C}\right)\right|$. Thus, $\left|\beta\left(\mathcal{C}\right)\right|\left|(\beta\left(\mathcal{C}\right))^{\perp}\right|=3^{n}$.

For the other direction, assume that $\beta\left(\mathcal{C}\right)\cap\left(\beta\left(\mathcal{C}\right)\right)^{\perp}=\left\{ \mathbf{0}\right\}$,
$\left|\beta\left(\mathcal{C}\right)\right|\left|(\beta\left(\mathcal{C}\right))^{\perp}\right|=3^{n}$ and $\left|\beta\left(\mathcal{C}\right)\right|=\left|\mathcal{C}\right|$.
We will show that $\mathcal{C}\cap\mathcal{C}^{\perp}=\left\{ \mathbf{0}\right\} $ and $\left|\mathcal{C}\right|\left|\mathcal{C}^{\perp}\right|=9^{n}$.
Since $\beta\left(\mathcal{C}\right)\cap\left(\beta\left(\mathcal{C}\right)\right)^{\perp}=\left\{ \mathbf{0}\right\}$, no non-zero element of $\beta\left(\mathcal{C}\right)$ is in
$\left(\beta\left(\mathcal{C}\right)\right)^{\perp}$. Also, since $\left|\beta\left(\mathcal{C}\right)\right|=\left|\mathcal{C}\right|$, there is only one non-zero element of $\mathcal{C}$ that corresponds to each non-zero element of $\beta\left(\mathcal{C}\right)$.
Suppose $\beta\left(\mathbf{c}_{1}\right)\in\beta\left(\mathcal{C}\right)$ is non-zero. We shall show that $\mathbf{c}_{1}\notin\mathcal{C}\cap\mathcal{C}^{\perp}$.
Since $\mathbf{c}_{1}\in\mathcal{C}$, we only need to show that $\mathbf{c}_{1}\notin\mathcal{C}^{\perp}$.
Since $\beta\left(\mathbf{c}_{1}\right)$ is non-zero and $\beta\left(\mathbf{c}_{1}\right)\notin\beta\left(\mathcal{C}\right)^{\perp}$, then there exists $\mathbf{c}_{2}\in\mathcal{C}$ such that $\langle \beta\left(\mathbf{c}_{1}\right) , \beta\left(\mathbf{c}_{2}\right) \rangle\neq0$.
Then it follows that $(\mathbf{c}_{1} , \mathbf{c}_{2}) = b$ or $2b\neq0$.
Thus, we have found a codeword $\mathbf{c}_{2}\in\mathcal{C}$ such that $(\mathbf{c}_{1} , \mathbf{c}_{2})\neq0$.
Therefore, $\mathbf{c}_{1}\notin\mathcal{C}^{\perp}$.
Thus, $\mathcal{C}\cap\mathcal{C}^{\perp}=\left\{ \mathbf{0}\right\}$.
Now, we have
$\left|\beta\left(\mathcal{C}\right)\right|
=\left|\mathcal{C}\right|
=3^{k}$ for some $k$. Thus,
$\left|(\beta\left(\mathcal{C}\right))^{\perp}\right|=3^{n-k}$.
As observed earlier, because 0 corresponds to 0 or $b$ or $2b$ and 1 to $a$ or $a+b$ or $a+2b$ and 2 to $2a$ or $2a+b$ or $2a+2b$, each element of $\beta\left(\mathcal{C}\right)^{\perp}$ corresponds to $3^{n}$ elements of $\mathcal{C}^{\perp}$. Thus, $\left|\mathcal{C}^{\perp}\right|=3^{2n-k}$. Therefore $\left|\mathcal{C}\right|\left|\mathcal{C}^{\perp}\right|=9^{n}$.
\end{proof}
}

The general case of the above theorem to any ring $I_{p}$ for $p$ prime bigger than 3 is given below.

\begin{theorem}  \label{thm: ACD iff LCD Ip-code}
An additive code $\mathcal{C}\subseteq I_{p}^{n}$ is ACD if and only if $\beta_{p}\left(\mathcal{C}\right)\subseteq\mathbb{F}_{p}^{n}$ is LCD {\color{blue}and $\left|\beta_{p}\left(\mathcal{C}\right)\right|=\left|\mathcal{C}\right|$.}
\end{theorem}

{\color{blue}
\begin{proof}
Assume that $\mathcal{C}\subseteq I_{p}^{n}$ is ACD, i.e., $\mathcal{C}\cap\mathcal{C}^{\perp}=\left\{ \mathbf{0}\right\} $ and  
$\left|\mathcal{C}\right|\left|\mathcal{C}^{\perp}\right|=p^{2n}$. We will show that $\beta_{p}\left(\mathcal{C}\right)\cap\left(\beta_{p}\left(\mathcal{C}\right)\right)^{\perp}=\left\{ \mathbf{0}\right\}$, $\left|\beta_{p}\left(\mathcal{C}\right)\right|\left|(\beta_{p}\left(\mathcal{C}\right))^{\perp}\right|=p^{n}$ and $\left|\beta_{p}\left(\mathcal{C}\right)\right|=\left|\mathcal{C}\right|$. Since we assume $\mathcal{C}\cap\mathcal{C}^{\perp}=\left\{ \mathbf{0}\right\}$, no non-zero element of $\mathcal{C}$ is in $\mathcal{C}^{\perp}$. 
Suppose $\mathbf{c}_{1}=\mathbf{c}_{i_{1}j_{1}}\in\mathcal{C}$ is non-zero. We may assume that $\beta_{p}\left(\mathbf{c}_{1}\right)$ is non-zero. If $\beta_{p}\left(\mathbf{c}_{1}\right)$ is zero, there is nothing to prove. We shall show that $\beta_{p}\left(\mathbf{c}_{1}\right)\notin\beta_{p}\left(\mathcal{C}\right)\cap\left(\beta_{p}\left(\mathcal{C}\right)\right)^{\perp}$. Since $\beta_{p}\left(\mathbf{c}_{1}\right)\in\beta_{p}\left(\mathcal{C}\right)$, we only need to show that $\beta_{p}\left(\mathbf{c}_{1}\right)\notin\beta_{p}\left(\mathcal{C}\right)^{\perp}$. 
Since $\mathbf{c}_{1}=\mathbf{c}_{i_{1}j_{1}}$ is non-zero and $\mathbf{c}_{1}\notin\mathcal{C}^{\perp}$, there is a non-zero $\mathbf{c}_{2}=\mathbf{c}_{i_{2}j_{2}}\in\mathcal{C}$ such that $\left(\mathbf{c}_{1},\mathbf{c}_{2}\right) \neq0$.
Since $\left(\mathbf{c}_{1},\mathbf{c}_{2}\right) \neq0$, we see that both $i_{1}$ and $i_{2}$ are non-zero. Then it follows that $\beta_{p}\left(\mathbf{c}_{2}\right)$ is non-zero and 
$\left\langle\beta_{p}\left(\mathbf{c}_{1}\right),\beta_{p}\left(\mathbf{c}_{2}\right)\right\rangle \neq0$. Thus, we have found a codeword $\beta_{p}\left(\mathbf{c}_{2}\right)\in\beta_{p}\left(\mathcal{C}\right)$ such that 
$\left\langle\beta_{p}\left(\mathbf{c}_{1}\right),\beta_{p}\left(\mathbf{c}_{2}\right)\right\rangle \neq0$. Therefore, $\beta_{p}\left(\mathbf{c}_{1}\right)\notin\left(\beta_{p}\left(\mathcal{C}\right)\right)^{\perp}$. 
Hence $\beta_{p}\left(\mathcal{C}\right)\cap\left(\beta_{p}\left(\mathcal{C}\right)\right)^{\perp}=\left\{ \mathbf{0}\right\}$.
 Since the additive group of $I_{p}$ is isomorphic to $(\mathbb{F}_{p})^{2}$, every additive subgroup of $I_{p}^{n}$ is naturally an $\mathbb{F}_{p}$-vector space. Hence, $|\mathcal{C}|=p^{k}$ for some integer $k$. Since
$|\mathcal{C}||\mathcal{C}^{\perp}|=p^{2n}$, it follows that
$
|\mathcal{C}^{\perp}|
=
p^{2n-k}$.
 Because 0 corresponds to 0, $b$, $2b$, $\dotsc$, $(p-1)b$, and 1 to $a$, $a+b$, $a+2b$, $\dotsc$, $a+(p-1)b$, and more generally, $k$ to $ka$, $ka+b$, $ka+2b$, $\dotsc$, $ka+(p-1)b$, each element of $\beta_{p}\left(\mathcal{C}\right)^{\perp}$ corresponds to $p^{n}$ elements of $\mathcal{C}^{\perp}$. Thus, $\left|(\beta_{p}\left(\mathcal{C}\right))^{\perp}\right|=p^{n-k}$.
We note that for $\mathcal{C}$ to be an ACD code, $\left|\mathcal{C}\right|$ has to be equal to $\left|\beta_{p}\left(\mathcal{C}\right)\right|$. 
We will show that $\left|\mathcal{C}\right|=\left|\beta_{p}\left(\mathcal{C}\right)\right|$. We note that since $\beta_{p}$ is a reduction map, $\left|\mathcal{C}\right|\geq\left|\beta_{p}\left(\mathcal{C}\right)\right|$. Suppose, for a contradiction, that $\left|\mathcal{C}\right|>\left|\beta_{p}\left(\mathcal{C}\right)\right|$. Then for a certain $\beta_{p}\left(\mathbf{c}_{1}\right)\in\beta_{p}\left(\mathcal{C}\right)$, there exist at least two distinct codewords $\mathbf{c}_{1},\mathbf{c}_{2}\in\mathcal{C}$ such that $\beta_{p}\left(\mathbf{c}_{1}\right)=\beta_{p}\left(\mathbf{c}_{2}\right)$. Thus, $(p-1)\beta_{p}\left(\mathbf{c}_{1}\right)+\beta_{p}\left(\mathbf{c}_{2}\right)=0$. Then $(p-1)\mathbf{c}_{1}+\mathbf{c}_{2}$ is non-zero and  all of its coordinates belong to $\{ 0 , b , 2b, \dotsc , (p-1)b \}$. This non-zero $(p-1)\mathbf{c}_{1}+\mathbf{c}_{2}$ is then in both
$\mathcal{C}$ and $\mathcal{C}^{\perp}$, which is a contradiction. Therefore, $\left|\mathcal{C}\right|=\left|\beta_{p}\left(\mathcal{C}\right)\right|$. Thus, $\left|\beta_{p}\left(\mathcal{C}\right)\right|\left|(\beta_{p}\left(\mathcal{C}\right))^{\perp}\right|=p^{n}$.

For the other direction, assume that $\beta_{p}\left(\mathcal{C}\right)\cap\left(\beta_{p}\left(\mathcal{C}\right)\right)^{\perp}=\left\{ \mathbf{0}\right\}$, 
$\left|\beta_{p}\left(\mathcal{C}\right)\right|\left|(\beta_{p}\left(\mathcal{C}\right))^{\perp}\right|=p^{n}$ and $\left|\beta_{p}\left(\mathcal{C}\right)\right|=\left|\mathcal{C}\right|$.
We will show that $\mathcal{C}\cap\mathcal{C}^{\perp}=\left\{ \mathbf{0}\right\} $ and $\left|\mathcal{C}\right|\left|\mathcal{C}^{\perp}\right|=p^{2n}$.
Since $\beta_{p}\left(\mathcal{C}\right)\cap\left(\beta_{p}\left(\mathcal{C}\right)\right)^{\perp}=\left\{ \mathbf{0}\right\}$, no non-zero element of $\beta_{p}\left(\mathcal{C}\right)$ is in 
$\left(\beta_{p}\left(\mathcal{C}\right)\right)^{\perp}$. Also, since $\left|\beta_{p}\left(\mathcal{C}\right)\right|=\left|\mathcal{C}\right|$, there is only one non-zero element of $\mathcal{C}$ that corresponds to each non-zero element of $\beta_{p}\left(\mathcal{C}\right)$.
Suppose $\beta_{p}\left(\mathbf{c}_{1}\right)\in\beta_{p}\left(\mathcal{C}\right)$ is non-zero. We shall show that $\mathbf{c}_{1}\notin\mathcal{C}\cap\mathcal{C}^{\perp}$. 
Since $\mathbf{c}_{1}\in\mathcal{C}$, we only need to show that $\mathbf{c}_{1}\notin\mathcal{C}^{\perp}$. 
Since $\beta_{p}\left(\mathbf{c}_{1}\right)$ is non-zero and $\beta_{p}\left(\mathbf{c}_{1}\right)\notin\beta_{p}\left(\mathcal{C}\right)^{\perp}$, then there exists $\mathbf{c}_{2}\in\mathcal{C}$ such that 
$\left\langle\beta_{p}\left(\mathbf{c}_{1}\right),\beta_{p}\left(\mathbf{c}_{2}\right)\right\rangle\neq0$. 
Then it follows that $\left(\mathbf{c}_{1},\mathbf{c}_{2}\right)=b$, $2b$, ..., or $(p-1)b\neq0$. 
Thus, we have found a codeword $\mathbf{c}_{2}\in\mathcal{C}$ such that $\left(\mathbf{c}_{1},\mathbf{c}_{2}\right)\neq0$. 
Therefore, $\mathbf{c}_{1}\notin\mathcal{C}^{\perp}$.
Thus, $\mathcal{C}\cap\mathcal{C}^{\perp}=\left\{ \mathbf{0}\right\} $. 
Now, we have $\left|\beta_{p}\left(\mathcal{C}\right)\right|
=\left|\mathcal{C}\right|=p^{k}$ for some $k$. Thus, $\left|(\beta_{p}\left(\mathcal{C}\right))^{\perp}\right|=p^{n-k}$.
As observed earlier, 
because 0 corresponds to 0, $b$, $2b$, $\cdots$, $(p-1)b$, and 1 to $a$, $a+b$, $a+2b$, $\cdots$, $a+(p-1)b$, and more generally, $k$ to $ka$, $ka+b$, $ka+2b$, $\cdots$, $ka+(p-1)b$, each element of $\beta_{p}\left(\mathcal{C}\right)^{\perp}$ corresponds to $p^{n}$ elements of $\mathcal{C}^{\perp}$. Thus, $\left| \mathcal{C}^{\perp} \right|=p^{2n-k}$. Therefore $\left|\mathcal{C}\right|\left|\mathcal{C}^{\perp}\right|=p^{2n}$.
\end{proof}
}

The next theorem classifies ACD codes for $n=1$ using \Cref{thm: ACD iff LCD I3-code} together with the classification of ternary LCD codes in \cite{Harada}.
Note that throughout this section, we only deal with codes $\mathcal{C}$ that satisfy the condition $\left|\beta\left(\mathcal{C}\right)\right|=\left|\mathcal{C}\right|$.
\begin{theorem}
For $n=1$, there are only four ACD codes over $I_{3}$.
\end{theorem}

\begin{proof}
We classify ACD codes of length 1 using \Cref{thm: ACD iff LCD I3-code} as follows. There are only two ternary LCD codes of length 1. One is $\left\{ 0\right\} $. The corresponding ACD code over $I_{3}$ is $\left\{ 0\right\} $. The other ternary LCD code is $\beta\left(\mathcal{C}\right)=\left\{ 0,1,2\right\} $. There are twenty seven codes of length 1 over $I_{3}$ that correspond to $\beta\left(\mathcal{C}\right)$ because each element of $\beta\left(\mathcal{C}\right)$ corresponds to 3 elements of $I_{3}$ and there are three elements in $\mathcal{C}$. Among these, $\left\{ 0,a,2a\right\} $, $\left\{ 0,a+b,2a+2b\right\} $  and $\left\{ 0,a+2b,2a+b\right\} $ are additive and ACD codes.
\end{proof}

In the following example, we show that a simple code of length 2 is ACD using \Cref{thm: ACD iff LCD I3-code}.
\begin{example}
For $\mathcal{C}=\left\{ 00, a0, 2a0\right\} $,
we note that
$\beta\left(\mathcal{C}\right)=\left\{ 00, 10, 20\right\}$ and
$\beta\left(\mathcal{C}\right)^{\perp}=\left\{ 00, 01,02\right\} $.
Then we have
$\left|\beta\left(\mathcal{C}\right)\right|\left|(\beta\left(\mathcal{C}\right))^{\perp}\right|=3^{2}$ and
$\beta\left(\mathcal{C}\right) \cap \beta\left(\mathcal{C}\right)^{\perp}=\left\{ \mathbf{0}\right\} $.
Hence, $\beta\left(\mathcal{C}\right)$ is a ternary LCD code. By \Cref{thm: ACD iff LCD I3-code}, $\mathcal{C}$ is an ACD code.
\end{example}

The next theorem classifies ACD codes for $n=2$.
\begin{theorem}
For $n=2$, there are exactly 40 ACD codes up to equivalence over $I_{3}$.
\end{theorem}

\begin{proof}
We divide the proof into three parts depending on the cardinality of a code.
\begin{enumerate}

\item[{(i)}] cardinality one
\\
$\left\{ 00\right\} $ is a trivial ACD code.
\item[{(ii)}] cardinality three
\\
According to \cite{Harada}, there are two ternary $[2,1]$ LCD codes up to equivalence, one of which is $\beta\left(\mathcal{C}\right)=\left\{ 00, 10, 20\right\} $. By \Cref{thm: ACD iff LCD I3-code}, there are 6 ACD codes $\mathcal{C}$ that correspond to this code, which are exactly $\left\{ 00,a0,(2a)(0)\right\} $, $\left\{ 00,ab,(2a)(2b)\right\} $, $\left\{ 00,(a+b)(0),\right.\\ \left.(2a+2b)(0)\right\} $, $\left\{ 00,(a+b)(b),(2a+2b)(2b)\right\} $, $\left\{ 00,(a+2b)(0), (2a+b)(0)\right\}$, \\
$\left\{ 00,(a+2b)(b),(2a+b)(2b)\right\} $.
The other ternary $[2,1]$ LCD code is $\beta\left(\mathcal{C}\right)=\left\{ 00, 12, 21\right\} $. By \Cref{thm: ACD iff LCD I3-code}, there are 6 ACD codes $\mathcal{C}$ corresponding to this code, namely, $\left\{ 00,(a)\right. \\ \left.(2a),(2a)(a)\right\} $,$\left\{ 00,(a)(2a+b),(2a)(a+2b)\right\} $, $\left\{ 00,(a)(2a+2b), (2a)(a+b)\right\} $, \\
$\left\{ 00,(a+b)(2a+b),(2a+2b)(a+2b)\right\} $, $\left\{ 00,(a+b) (2a+2b),(2a+2b)(a+b)\right\} $, \\
$\left\{ 00,(a+2b)(2a+b),(2a+b)(a+2b)\right\} $.

\item[{(iii)}] cardinality nine
\\
The next possible cardinality of an additive code is nine because the additive code is a subgroup of $I_{3}^{2}$.
Since it is not straightforward to calculate the number of monomially inequivalent ACD codes by hand, we make use of the Magma for the calculation. According to the Magma computation, there are 27 ACD codes of this cardinality.
\end{enumerate}
This completes the proof.
\end{proof}

There is a trivial ACD code with cardinality one for $n=3$, namely $\left\{ 000\right\} $.
The next theorem counts ACD codes with cardinality three for $n=3$.
\begin{theorem}
For $n=3$, there are exactly 21 ACD codes with cardinality three up to equivalence.
\end{theorem}

\begin{proof}
According to \cite{Harada}, there are only two ternary $[3,1]$ LCD codes up to equivalence, one of which is $\beta\left(\mathcal{C}\right)=\left\{ 000, 110, 220\right\} $. By \Cref{thm: ACD iff LCD I3-code}, there are $3^{3}-9-6=12$ ACD codes over $I_{3}$ corresponding to this ternary LCD code. From the twenty seven we take away nine because $(a+b)(a)(0)$ is the same as $(a)(a+b)(0)$ when permuted together with the cases where $0$ is replaced by $b$ or $2b$. The same reasoning applies to $(a+2b)(a)(0)$ and $(a)(a+2b)(0)$, and $(a+2b)(a)(0)$ and $(a)(a+2b)(0)$. We also take away six for the following reason. Since 110 is the generator of the ternary LCD code, without loss of generality, let us consider the elements of $I_{3}^{3}$ that correspond to 110 only. We notice that $aab$ is monomially equivalent to $(a)(a)(2b)$. Similary for the cases $(a)(a+b)(b)$, $(a)(a+2b)(b)$, $(a+b)(a+b)(b)$, $(a+b)(a+2b)(b)$, $(a+2b)(a+2b)(b)$.
According to \cite{Harada}, there is another inequivalent ternary $[3,1]$ LCD code $\beta\left(\mathcal{C}\right)=\left\{ 000, 100, 200\right\} $. By \Cref{thm: ACD iff LCD I3-code}, there are $3^{3}-9-9=9$ ACD codes over $I_{3}$ corresponding to this ternary LCD code. From the twenty seven we take away nine because $a0b$ is the same as $ab0$ when permuted, $(a)(0)(2b)$ the same as $(a)(2b)(0)$, $(a)(b)(2b)$ the same as $(a)(2b)(b)$. Similarly for the cases where $a$ is replaced by $a+b$ or $a+2b$. We also take away three because $a0b$ is monomially equivalent to $(a)(0)(2b)$, and similary for the cases where $a$ is replaced by $a+b$ or $a+2b$. We additionally take away six because $abb$ is monomially equivalent to $(a)(b)(2b)$ and $(a)(2b)(2b)$, and similarly for the cases where $a$ is replaced by $a+b$ or $a+2b$. Thus, in total we have 21 inequivalent ACD codes. A Magma computation confirms that this result is correct.
\end{proof}

The next theorem counts ACD codes with cardinality nine for $n=3$.
\begin{theorem}
For $n=3$, there are 333 ACD codes with cardinality nine up to equivalence.
\end{theorem}

\begin{proof}
Since calculating by hand is not easy, we have used Magma for the computation.
Using the exhaustive search, we have checked that there are exactly 333 ACD codes with cardinality nine up to equivalence.
\end{proof}

The classification of ACD codes over $I_{3}$ is summarized in Table 6. For each length and cardinality, the number of ACD codes is given in the third column. The highest  minimum Hamming and Lee distances are given in the next two columns. The last column provides examples of ACD codes with the highest Lee and hence Hamming minimum distances. MDS over $I_{3}$ is marked with $\dagger$ and ternary optimal LCD code is marked with $^*$. We note that in some cases, we have MDS over $I_{3}$, and in all cases we have ternary optimal LCD code.

\begin{table}[H]
\begin{center}

\caption{Summary of classifications over $I_{3}$. Here $\dagger$ denotes MDS over $I_3$ and * denotes a ternary optimal LCD code}\label{tab1}%

\resizebox{\textwidth}{!}{%
\begin{tabular}{p{0.3cm}>{\centering}p{2.1cm}>{\centering}p{2.8cm}>{\centering}p{2.3cm}>{\centering}p{2.3cm}>{\centering}p{4.3cm}}
\hline
$n$ & Cardinality  &  \#ACD codes &  Highest $d_H$($d_L$) & Ternary LCD & Example \tabularnewline
\hline
1 & 3 & 3 & 1($2^*$) & $\left[1,1,1\right]^{*}$ & $\left\{ 0,a+2b,2a+b\right\} $\tabularnewline

2 & 3 & 12 & 2($4^*$) & $\left[2,1,2\right]^{*}$ & $\{ 00,(a+2b)(2a+b),$\\
                                                 $(2a+b)(a+2b)\} $  \tabularnewline

2 & 9 & 27 & $2^{\dagger}$($3^*$) & $\left[2,2,1\right]^{*}$ &  $\{ 00,(2a+2b)(a),$\\
                                                   $(a+b)(2a),$\\
                                                  $(2b)(2a+2b),(b)(a+b),$\\
                                                  $(2a)(2a+b),(a)(a+2b),$\\
                                                  $(a+2b)(b),(2a+b)(2b)\}$     \tabularnewline

3 & 3 & 21 & 3($5^*$) & $\left[3,1,2\right]^{*}$ & $\{ 000,(a+2b)(a+2b)(2b),$\\
                                                 $(2a+b)(2a+b)(b)\} $ \tabularnewline

3 & 9 & 333 & $3^{\dagger}$($4^*$) & $\left[3,2,1\right]^{*}$ & $\{ 000,(2a)(b)(a+2b),$\\
                                                  $(a+2b)(2a)(2a+2b),$\\
                                                  $(a)(2b)(2a+b),$\\
                                                  $(2a+2b)(2a+2b)(a),$\\
                                                  $(a+b)(a+b)(2a),$\\
                                                  $(b)(a+2b)(2b),$\\
                                                  $(2b)(2a+b)(b),$\\
                                                  $(2a+b)(a)(a+b)\}$ \tabularnewline
\hline
\end{tabular}}

\end{center}
\end{table}

\section{Conclusion}
\label{sec: conclusion}

We have introduced additive complementary dual (ACD) codes over the rings $I_{2}$ and $I_{3}$ since we have observed that there are no nontrivial LCD codes over $I_{2}$  and $I_{3}$. We have shown relations between ACD codes over $I_{2}$ and binary LCD codes and between ACD codes over $I_{3}$ and ternary LCD codes. Using the first relation and Magma computation, we have classified ACD codes over $I_{2}$ with the highest minimum distances for $n=1, 2, 3$ and partially for $n=4, 5$. Using the second relation and Magma computation, we have classified ACD codes over $I_{3}$ with the highest minimum distances for $n=1, 2$ and partially for $n=3$. We have generalized the two relations into a relation between ACD codes over $I_{p}$ and $p$-ary LCD codes. Using the relations simplifies the classifications for small lengths and cardinalities over $I_{2}$ and $I_{3}$, but still exhaustive search and Magma computations are needed as length and cardinality increase because there are no apparent patterns among ACD codes. Hence, as future work one can construct or classify ACD codes over $I_{2}$ or $I_{3}$ of length $n \ge 4$.

\noindent
As an example of future work, one may modify the construction methods for self-orthogonal codes introduced in \cite{Kim5}. Specifically, if the generator matrix of $\alpha\left(\mathcal{C}\right)\subseteq\mathbb{F}_{2}^{n}$ is given by $G=\left(xI,yM\right)$ where $x,y\in\mathbb{F}_{2}$, $I_{2}$ the identity matrix, $M$ a binary matrix, in analogy with a construction method in \cite{Kim5}, we can determine whether $\alpha\left(\mathcal{C}\right)$ is LCD or not as follows. For $x=y=0$, $\alpha\left(\mathcal{C}\right)$ cannot be LCD. If $x=0$ and $y\neq0$, $\alpha\left(\mathcal{C}\right)$ is LCD if det$M=1$. When $x\neq0$ and $y=0$, $\alpha\left(\mathcal{C}\right)$ is always LCD. For $x\neq0$ and $y\neq0$, $\alpha\left(\mathcal{C}\right)$ is LCD if det$(I+MM^{T})=1$. For those cases where $\alpha\left(\mathcal{C}\right)$ is LCD, $\mathcal{C}\subseteq I_{2}^{n}$ is ACD using \Cref{thm: ACD iff LCD I2-code}.

\section*{Acknowledgments}
J.-L. Kim was supported in part by the BK21 FOUR (Fostering Outstanding Universities for Research) funded by the Ministry of Education (MOE, Korea) and National Research Foundation of Korea (NRF) under Grant No. 4120240415042 and by Basic Science Research Program through the National Research Foundation of Korea (NRF) funded by the Ministry of Science and ICT under Grant No. RS-2025-24534992.
Y.G. Roe was supported by Basic Science Research Program through the National Research Foundation of Korea (NRF) funded by the Ministry of Education (RS-2025-25415913).









\end{document}